\documentclass[aps,pre,twocolumn,superscriptaddress,showpacs]{revtex4-1}

\usepackage{graphicx}
\usepackage{epsfig}
\usepackage{subfigure}
\usepackage{amssymb,amsfonts,amsmath,amsthm,amssymb,mathtools,mathrsfs}

\usepackage{times}

\newtheorem{mythm}{Theorem}
\newtheorem{mylma}{Lemma}
\newtheorem{mydef}{Definition}

\begin{document}

\title{Generalized Measures of Information Transfer}

\author{Paul L. Williams}
\affiliation{Cognitive Science Program and}
\author{Randall D. Beer}
\affiliation{Cognitive Science Program and}
 \affiliation{School of Informatics and Computing\\ Indiana University, Bloomington, Indiana 47406 USA}

\date{\today}

\begin{abstract}
Transfer entropy provides a general tool for analyzing the magnitudes and directions---but not the \emph{kinds}---of information transfer in a system.  We extend transfer entropy in two complementary ways.  First, we distinguish state-dependent from state-independent transfer, based on whether a source's influence depends on the state of the target. Second, for multiple sources, we distinguish between unique, redundant, and synergistic transfer.  The new measures are demonstrated on several systems that extend examples from previous literature.
\end{abstract}


\maketitle

\section{Introduction}

Many scientific problems involve understanding the behavior of complex systems in terms of the interactions between their component parts.
Using information theory, these interactions can be quantified as the information exchanged between components \cite{Hlavackova2007,*Lungarella2007}.  Such an approach has the advantages that informational measures are sensitive to arbitrary nonlinear interactions between components and have units of measurement (bits) that are easily interpreted and compared across systems.  In particular, the information-theoretic measure \emph{transfer entropy} (TE) \cite{Schreiber2000,Kaiser2002} has become widely adopted as a standard measure of information transfer, with applications in neuroscience \cite{Staniek2008,*Honey2007,*Gourevitch2007}, cellular biology \cite{Pahle2008}, chaotic synchronization \cite{Hung2008,*Otsuka2002,*Palus2001}, and econophysics \cite{Kwon2008, *Marschinski2002} to name just a few. 
	
Transfer entropy provides a directional measure of the influence that one random process, the \emph{source}, has on another, the \emph{target}.  This influence is measured by the information that the source provides about the next state of the target when conditioned on the target's history.  The idea behind TE is that conditioning on the target's history removes the information shared by the source and target due to common histories or inputs, thereby isolating the information that is actually transferred. However, conditioning does not simply remove shared information; it also adds in higher-order synergistic information, an idea that was formalized in the recently proposed \emph{partial information (PI) decomposition} \cite{Williams2010}. 

Here we apply this basic property of conditional information to generalize TE in two complementary ways.  First, we decompose TE into two kinds of information transfer that differ regarding the influence of the target's state.  We show that the resulting measures are formally related to the control-theoretic concepts of open-loop and closed-loop control, and quantify separately the state-independent and state-dependent influences of the source onto the target.  Second, we apply a similar decomposition to the case of multiple sources and derive a novel multivariate generalization of TE.  The resulting measures quantify separately the unique, redundant, and synergistic influences of multiple sources onto a target.  Together these results provide a general framework for characterizing not only the magnitudes and directions but also the \emph{kinds} of information exchange that occur between random processes. 

We begin by introducing PI-decomposition for a system of three random variables, $X$, $Y$, and $Z$, for which we ask: How much total information do $Y$ and $Z$ provide about $X$? And, how do $Y$ and $Z$ contribute to the total information?  The answer to the former is given by the mutual information (MI)
\begin{equation}
I(X; Y, Z) = H(X) - H(X|Y, Z)
\end{equation}
where $H(X) = -\sum_x p(x) \log p(x)$ is the familiar Shannon entropy \cite{Cover2006}.  For the latter, we can identify three distinct possibilities, that is, three \emph{kinds} of information that $Y$ and $Z$ may provide.  First, $Y$ may provide information that $Z$ does not, or vice versa (\emph{unique information}).  For example, if $Y$ is a copy of $X$ and $Z$ is a degenerate random variable, then the total information reduces to the unique information from $Y$.  Second, $Y$ and $Z$ may provide the same or overlapping information (\emph{redundancy}).  For example, if $Y$ and $Z$ are both copies of $X$ then they redundantly provide complete information.  Third, the combination of $Y$ and $Z$ may provide information that is not available from either alone (\emph{synergy}).  The classic example for binary variables is the exclusive-OR function $X = Y \oplus Z$, in which case $Y$ and $Z$ individually provide no information but together provide complete information.  Thus, intuitively, $I(X;Y,Z)$ decomposes into unique information from $Y$ and $Z$, redundant information shared by $Y$ and $Z$, and synergistic information contributed jointly by $Y$ and $Z$.

PI-decomposition formalizes this idea, starting with a measure of redundancy.  Letting ${\bf R} = \{Y,Z\}$, redundancy is defined as
\begin{equation}\label{IminEq}
I_{\min}(X; Y, Z) = \sum_x p(x) \min_{R \in {\bf R}} I(X=x; R)
\end{equation}
where
\begin{equation}\label{SpecInfoEq}
I(X=x; R) = \sum_r p(r|x) \bigg[ \log \frac{1}{p(x)} - \log \frac{1}{p(x|r)} \bigg]
\end{equation}
is the specific information that $R$ provides about each state $X=x$.  Thus, redundancy is defined as the minimum information that $Y$ or $Z$ provides about each state of $X$, averaged over all possible states.  This definition captures the idea that redundancy is the information shared by $Y$ and $Z$ (the minimum that either provides) while taking into account that $Y$ and $Z$ may provide information about different states of $X$.

Using $I_{\min}$ and the inclusion-exclusion principle \cite{Stanley1997}, the total information $I(X; Y, Z)$ can then be decomposed into partial information terms, given by the PI-function $\Pi_{{\bf R}}$.  The redundancy is given by $\Pi_{{\bf R}}(X; \{Y\}\{Z\}) = I_{\min}(X; Y, Z)$.  The unique information from $Y$ is given by $\Pi_{{\bf R}}(X; \{Y\}) = I(X; Y) - I_{\min}(X; Y, Z)$, or the total information from $Y$ minus the redundancy, and likewise for $Z$.  Finally, the synergy is given by $\Pi_{{\bf R}}(X; \{Y,Z\}) = I(X; Y, Z) - I_{\max}(X; Y, Z)$, where $I_{\max}$ is defined the same as $I_{\min}$ except substituting $\max$ for $\min$.  Together these terms yield the decomposition:
\begin{align}
I(X; Y) & = \Pi_{{\bf R}}(X; \{Y\}) + \Pi_{{\bf R}}(X; \{Y\}\{Z\}) \label{MIDecompEq1} \\
& \mbox{and} \notag\\
I(X; Y, Z) & = \Pi_{{\bf R}}(X; \{Y\}) + \Pi_{{\bf R}}(X; \{Z\}) \notag \\
				& + \Pi_{{\bf R}}(X; \{Y\}\{Z\}) + \Pi_{{\bf R}}(X; \{Y,Z\}) \label{MIDecompEq2}.
\end{align}

\begin{figure}[t!]
  \centering
    \includegraphics[width=0.55\columnwidth]{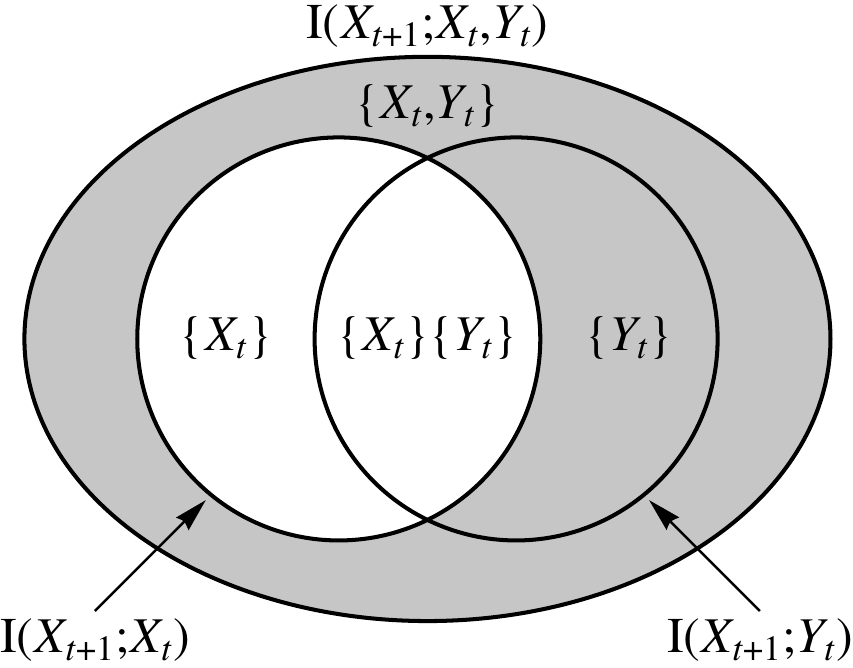} 
\caption{PI-decomposition of $I(X_{t+1}; X_t, Y_t)$ into unique information from $X_t$ ($\{X_t\}$) and $Y_t$ ($\{Y_t\}$), redundancy ($\{X_t\}\{Y_t\}$), and synergy ($\{X_t, Y_t\}$).  The gray region corresponds to $T_{Y \rightarrow X}$, which decomposes into SITE ($\{Y_t\}$) and SDTE ($\{X_t, Y_t\}$).}
\label{te-pi-decomposition}
\end{figure}

An immediate consequence of Eqs.~\eqref{MIDecompEq1} and \eqref{MIDecompEq2} is that the conditional MI $I(X;Y|Z)$ decomposes into the unique information from $Y$ \emph{plus} the synergy from $Y$ and $Z$:
\begin{align}\label{CondMIDecompEq}
I(X;Y|Z) & = I(X;Y,Z) - I(X;Z) \notag \\
	     & = \Pi_{{\bf R}}(X; \{Y\}) + \Pi_{{\bf R}}(X; \{Y,Z\}).
\end{align}
Thus, conditioning $I(X; Y)$ on $Z$ not only removes the redundancy from $Y$ and $Z$, but also adds in their synergy.  This observation makes intuitive sense if we think of $I(X;Y|Z)$ as answering the question: How much information do we gain from learning $Y$ when we already know $Z$?  Clearly, this will include both the information that comes uniquely from $Y$ plus the synergistic information that comes from $Y$ and $Z$ together.

Transfer entropy can be analyzed in a similar way, since it is simply an application of conditional MI to stochastic processes.  Given processes $X$ and $Y$, the TE from $Y$ to $X$ is defined as
\begin{equation}\label{TEEq}
T_{Y \rightarrow X} = I(X_{t+1}; Y_t^{(l)} | X_t^{(k)})
\end{equation}
where $X_t^{(k)}$ is the $k$-dimensional delay vector for $X$, and likewise for $Y_t^{(l)}$ and $Y$ (henceforth the superscripts are omitted for clarity).  In other words, $T_{Y \rightarrow X}$ quantifies the information that previous values of $Y$ provide about the next state of $X$ when conditioned on $X$'s own history.  In terms of transition probabilities, $T_{Y \rightarrow X}$ can also be thought of as quantifying deviation from the generalized Markov property $p(x_{t+1}|x_t, y_t) = p(x_{t+1}|x_t)$, with $T_{Y \rightarrow X} = 0$ iff $Y$ has no influence on the transitions of $X$.

\section{Decomposing Transfer Entropy}

Our first main result is that, by decomposing $T_{Y \rightarrow X}$, we can distinguish two kinds of information transfer; that is, two distinct ways that $Y$ can influence the transitions of $X$ (FIG.~\ref{te-pi-decomposition}). Letting ${\bf R} = \{ X_t , Y_t \}$ and combining Eqs.~\eqref{CondMIDecompEq} and \eqref{TEEq}, we have that
\begin{equation}\label{TEDecompEq}
T_{Y \rightarrow X} = \Pi_{{\bf R}}(X_{t+1}; \{ Y_t \}) + \Pi_{{\bf R}}(X_{i+1}; \{ X_t, Y_t \})
\end{equation}
where $\Pi_{{\bf R}}(X_{t+1}; \{ Y_t \})$ is the unique information that $Y_t$ provides about $X_{t+1}$ and $\Pi_{{\bf R}}(X_{t+1}; \{ X_t, Y_t \})$ is the synergistic information from $X_t$ and $Y_t$. As we will show, $\Pi_{{\bf R}}(X_{t+1}; \{ Y_t \})$ corresponds to \emph{state-independent transfer entropy} (SITE): it measures the portion of $Y_t$'s influence on $X_{t+1}$ that does not depend on $X_t$.  The complementary term $\Pi_{{\bf R}}(X_{t+1}; \{ X_t, Y_t \})$ is the \emph{state-dependent transfer entropy} (SDTE): it measures the influence that $Y_t$ has on $X_{t+1}$ only when combined with an appropriate state of $X_t$. To ground this interpretation, we next establish a formal connection between SITE and SDTE and the control-theoretic notions of open-loop and closed-loop control.

In control theory, one considers a process $X_t$---characterized by its initial state $X$ and final state $X'$---and a controller $C$, with the two related by a distribution $p(x'|x, c)$ \cite{Touchette2000,Touchette2004}. The aim is to specify a control policy, given by the distribution $p(c|x)$, that moves the system to certain desired final states.  In open-loop control, the controller $C$ acts independently of the initial state $X$ ($I(X;C) = 0$), while closed-loop control is characterized by state-dependent actuation.

A fundamental property of a control system is its \emph{controllability}, which is the extent to which it can be moved through its entire state space.  In particular, a system has perfect controllability iff there is a control policy that moves the system deterministically from any $x \in X$ to any $x' \in X'$. In \cite{Touchette2004}, it is shown that a natural information-theoretic measure of controllability is $I(X'; C|X)$---the \emph{information transfer} from the controller to the controlled process---which is maximal exactly in the case of perfect controllability.  Thus, there is a close parallel between information transfer and controllability, where essentially the only difference is semantic: information transfer applies to arbitrary interactions between processes, while controllability is concerned specifically with using one process to influence another.

With this in mind, the following result connects SITE and SDTE with open-loop and closed-loop control (proof in Appendix~\ref{app:proofs}).
\begin{mythm}
A system is perfectly controllable with open-loop control iff it is perfectly controllable with only state-independent transfer from $C$ to $X'$.
\end{mythm}
Thus, decomposing $I(X';C|X)$ as in Equation \eqref{TEDecompEq}, SITE from $C$ to $X'$ measures a system's open-loop controllability (maximal for perfect open-loop control), while SDTE measures the additional contribution from closed-loop control. More generally, this connection grounds the interpretation of SITE as the state-independent (open-loop) influence of one process on another, and likewise for SDTE and state-dependent (closed-loop) influence.

\begin{figure}[t!]
  \centering
    \includegraphics[width=0.55\linewidth]{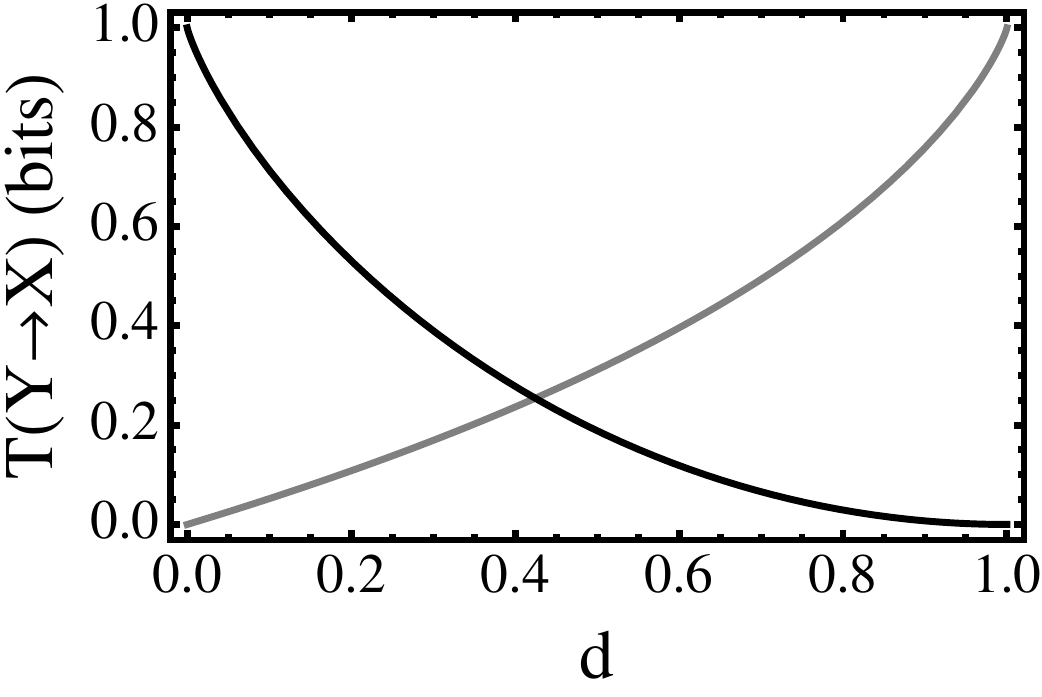}
\caption{SITE (black) and SDTE (gray) for binary Markov processes $X$ and $Y$ as a function of the coupling parameter $d$.}
\label{SDTEvsSITE}
\end{figure}

As a simple example to illustrate the two kinds of transfer, consider two binary state Markov processes $X$ and $Y$, where $Y$ is purely random and $X$ is stochastically coupled to $Y$.  Specifically, if $x_t = 0$, then $x_{t+1} = y_{t}$, while if $x_t = 1$, the probability that $x_{t+1} = y_t$ is $1-d$ and that $x_{t+1} = 1-y_t$ is $d$.  A simple eigenvector calculation yields the stationary distribution $p(x,y)=1/4$ for all $x$ and $y$, and from this all informational quantities can be computed.  When $d=0$, $X_{t+1}$ is simply set to $Y_t$ regardless of its own previous state, thus corresponding to pure SITE (FIG.~\ref{SDTEvsSITE}).  With this parameter setting, the system is essentially equivalent to the discrete example considered in \cite{Kaiser2002}.  In contrast, when $d=1$, $y_t=0$ causes $X$ to remain in the same state and $y_t = 1$ causes $X$ to switch states.  Consequently, $Y$'s influence on $X_{t+1}$ depends entirely on $X_t$, corresponding to pure SDTE.  In fact, if one imagines using $Y$ to control $X$, then $d=1$ corresponds to a `controlled-NOT' gate, which is known to require closed-loop control \cite{Touchette2000,Touchette2004}.  FIG.~\ref{SDTEvsSITE} shows how varying $d$ produces a smooth transition between these two extremes.

Finally, we note that the distinction between SITE and SDTE also clarifies the relationship between TE and the time-delayed mutual information (TDMI) $I(X_{t+1}; Y_t)$, which was the standard measure of information transfer prior to TE \cite{Schreiber2000}.  Transfer entropy was initially proposed as an alternative to TDMI because the latter fails to remove shared information due to common histories or inputs.  From FIG.~\ref{te-pi-decomposition}, it is clear that this shared information corresponds to $\Pi_{{\bf R}}(X_{t+1}; \{X_t\}\{Y_t\})$, the redundancy between $X_t$ and $Y_t$.  However, FIG.~\ref{te-pi-decomposition} also reveals a second crucial difference between TDMI and TE, which is that TDMI fails to include SDTE.  Thus, not only does TDMI incorrectly add in shared information, but it also leaves out a significant component of information transfer.

\begin{figure}[b!]
  \centering
    \includegraphics[width=0.75\columnwidth]{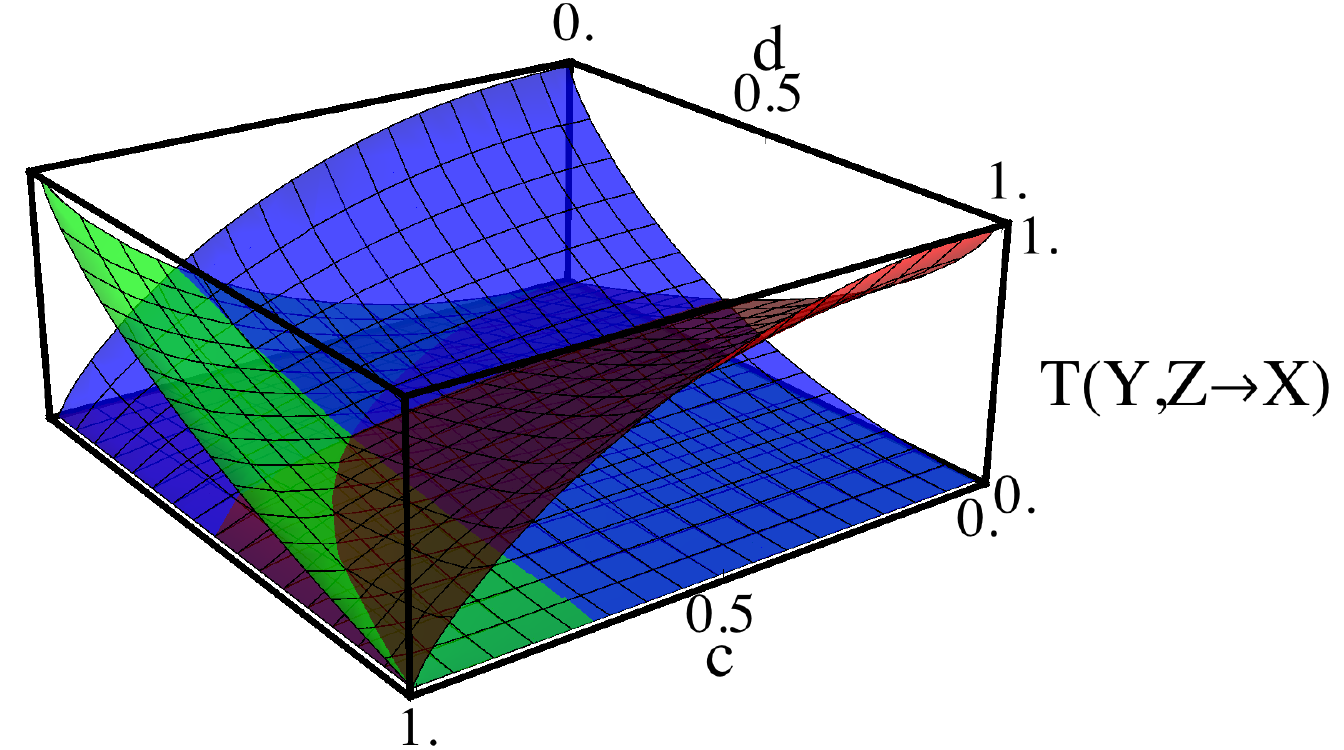} 
\caption{$T_{Y \rightarrow X \backslash Z}$ (blue), $T_{\{Y\}\{Z\} \rightarrow X}$ (green), and $T_{\{Y, Z\} \rightarrow X}$ (red) for binary Markov processes X, Y, and Z.}
\label{mvte_plot}
\end{figure}

\section{Multivariate Information Transfer}

Our second main result is a novel multivariate generalization of TE, based on applying PI-decomposition to the information from multiple sources.  Schreiber \cite{Schreiber2000} originally proposed a generalization of TE based on `conditioning out' other sources, an idea that has since been adopted and extended by others \cite{Frenzel2007,*Lizier2008}. However, it should be clear from the preceding discussion that such a generalization is problematic, since conditioning does not simply remove shared information.  Our generalization addresses this deficiency by quantifying separately the unique, redundant, and synergistic transfer from multiple sources.

For simplicity, we consider only two sources $Y$ and $Z$ acting on a target $X$ (the general case is discussed momentarily), in which case the total TE is given by $T_{Y,Z \rightarrow X} = I(X_{t+1}; Y_t, Z_t | X_t)$.    
Applying PI-decomposition as before, we arrive at measures for the \emph{redundant transfer} from $Y$ and $Z$: $T_{\{Y\}\{Z\} \rightarrow X} = I_{\min}(X_{t+1}; Y_t, Z_t | X_t)$; the \emph{unique transfer} from Y (resp. Z): $T_{Y \rightarrow X \backslash Z} = T_{Y \rightarrow X} - T_{\{Y\}\{Z\} \rightarrow X}$; and the \emph{synergistic transfer} from $Y$ and $Z$: $T_{\{Y, Z\} \rightarrow X} = T_{Y,Z \rightarrow X} - I_{\max}(X_{t+1}; Y_t, Z_t | X_t)$.  Generally speaking, redundant transfer corresponds to situations where the apparent influence from multiple sources may in fact be due to any one (or several) of them, indicating that interventional methods are required to determine the true causal structure \cite{Ay2008}.  In contrast, unique transfer represents the portion of a source's influence that can only come from that source, or, if all possible sources are considered, that \emph{must} come from that source.  Finally, synergistic transfer indicates that several sources act together cooperatively to influence the target.

To illustrate, consider three binary state Markov processes $X$, $Y$, and $Z$.  $Y$ is purely random, and $Z$ is stochastically coupled to $Y$ such that $z_t = y_t$ with probability $(1+c)/2$ and $z_t = 1-y_t$ with probability $(1-c)/2$.  This coupling can be thought of as an external signal driving $Y$ and $Z$ to synchronize:  as $c$ goes from $0$ to $1$, $Y$ and $Z$ transition from independence to complete synchronization. $X$ in turn is coupled to both $Y$ and $Z$ such that, if $z_t = 0$, $x_{t+1} = y_t$, while if $z_t = 1$, $x_{t+1} = y_t$ with probability $(1-d)$ and $x_{t+1} = (1-y_t)$ with probability $d$.    Thus, $x_{t+1} = y_t$ when $d=0$, and $x_{t+1} = y_t \oplus z_t$ when $d=1$.  At the extreme parameter settings, this system exhibits three different behaviors (FIG.~\ref{mvte_plot}).   With $(c=0, d=0)$, $Y$ and $Z$ are independent and $X$ depends only on $Y$, so the only influence is unique transfer from $Y$ to $X$.  In contrast, with $(c=1, d=0)$, $X$ again depends only on $Y$ but $Y$ and $Z$ are now synchronized, so there is only redundant transfer from $Y$ and $Z$.  Indeed, in this case it is impossible to determine from observation alone whether $Y$ or $Z$ (or both) is driving $X$.  Finally, with $(c=1, d=0)$, $Y$ and $Z$ are independent and $X_{t+1} = Y_t \oplus Z_t$, corresponding to pure synergistic transfer.

\begin{figure}[t]
  \centering
    \includegraphics[width=0.55\columnwidth]{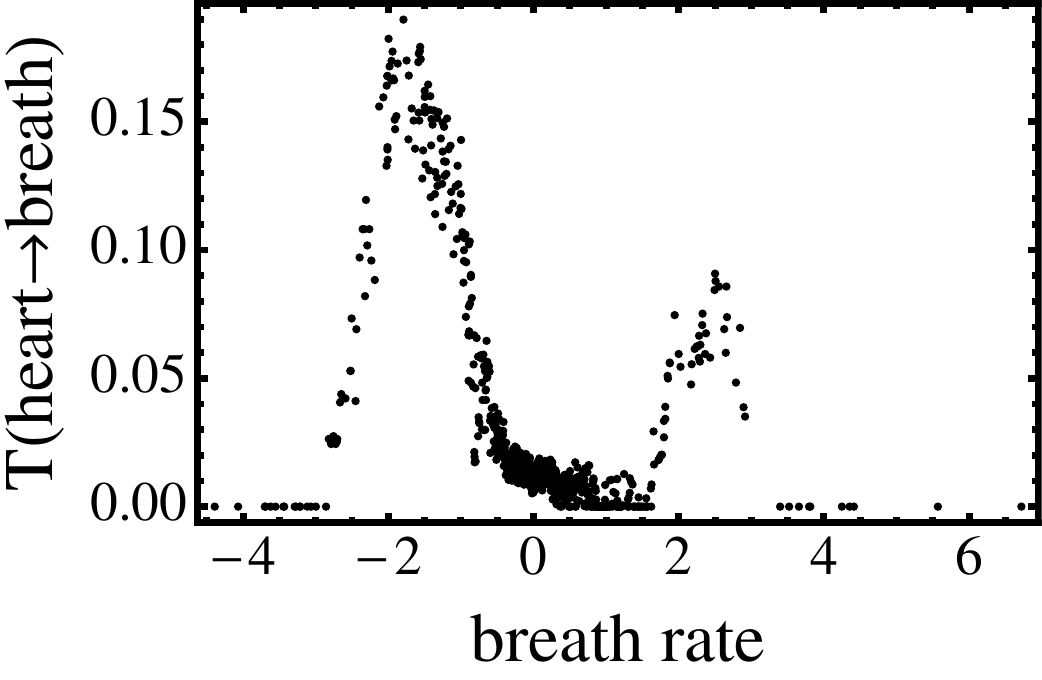} 
\caption{Information transfer $I(B_{t+1}; H_t |B_t = b_t)$ from heart rate ($H$) to breath rate ($B$) as a function of $b_t$ for bandwidth $r = 0.5$. Qualitatively similar results were found for $r \in [0.2,1.0]$.}
\label{physio_state_dep_plot}
\end{figure}

As a final example, we extend the analysis of a multivariate physiological time series presented in \cite{Schreiber2000, Kaiser2002}.  The data consists of simultaneous recordings of the breath rate (chest volume), heart rate, and blood oxygen concentration for a patient suffering from sleep apnea.  Previous analysis compared TE and TDMI for both directions between the breath and heart signals.  However, directly comparing TE and TDMI is problematic and potentially misleading, since both measures detect SITE but differ regarding SDTE and shared information (FIG.~\ref{te-pi-decomposition}). Indeed, with no additional information, it is impossible to determine even whether TE and TDMI are detecting the same or different aspects of an interaction.

\begin{figure}[b!]
  \centering
    \includegraphics[width=0.55\columnwidth]{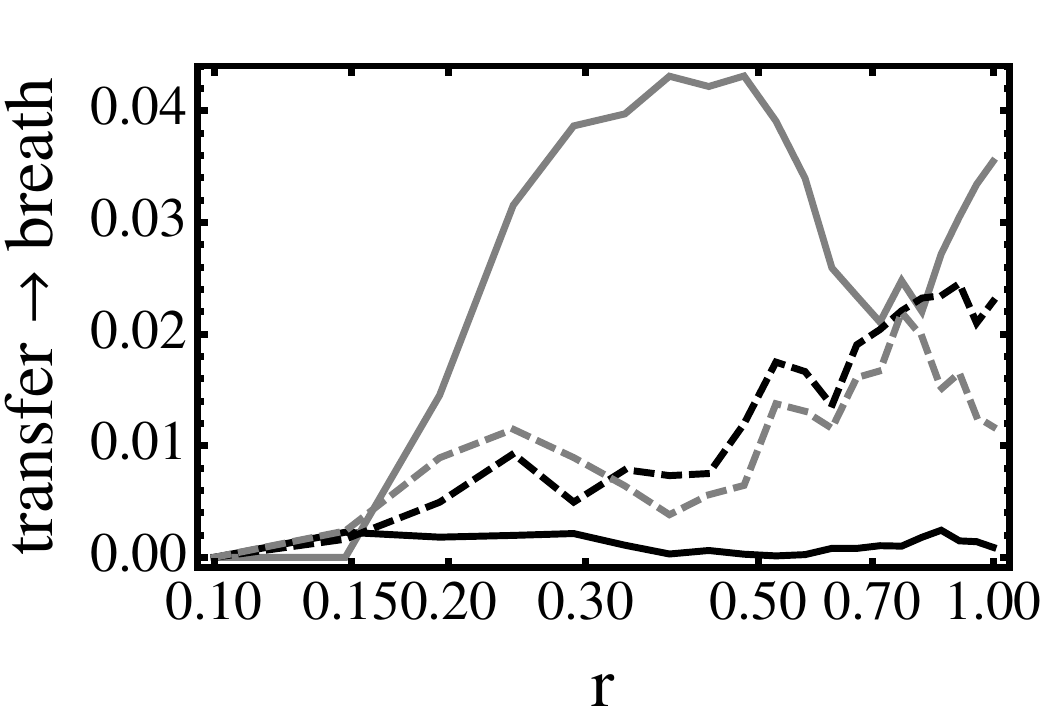} 
\caption{$T_{O \rightarrow B \backslash H}$ (black), $T_{H \rightarrow B \backslash O}$ (gray), $T_{\{H\}\{O\} \rightarrow B}$ (dashed black), and $T_{\{H, O\} \rightarrow B}$ (dashed gray) from heart rate (H) and blood oxygen (O) to breath rate (B).}
\label{physio_mvte_plot}
\end{figure}

To address this issue, we calculated SITE and SDTE between the breath and heart signals.  Joint probability estimates were obtained by kernel estimation using a rectangular kernel with bandwidth $r$.  Neighboring points closer than $20$ time steps were excluded and points with fewer than $5$ neighbors were ignored, following the suggestions of Schreiber and others \cite{Schreiber2000, Kaiser2002, Pahle2008}. Our main finding is that SITE is consistent with zero in both directions between the breath and heart signals.  Thus, for these signals, TE and TDMI in fact quantify entirely separate things: TDMI is due only to shared information from common histories or inputs, while TE detects state-dependent information exchange.  This state dependence can be seen by plotting information transfer as a function of the target state, shown in FIG~\ref{physio_state_dep_plot} for $T(\text{heart} \to \text{breath})$.  For pure SITE, this plot would be uniform across target states, while FIG.~\ref{physio_state_dep_plot} shows a clear bimodal distribution.  These two modes correspond to downswings and upswings in chest volume, suggesting that heart rate has the largest influence on respiration when chest volume is low and, to a lesser extent, when it is high, but minimal influence when chest volume is near its mean (see also FIG.~\ref{supp_physio_sdte_plot} in Appendix~\ref{app:fig}).   

We also analyzed the combined influence of heart rate and blood oxygen level on breath rate (FIG.~\ref{physio_mvte_plot}).  The most significant component is consistently the unique information transfer from the heart rate, indicating that most of the TE discussed above is uniquely attributable to the heart signal.  However, there is also considerable redundant and synergistic transfer, of roughly comparable magnitude, from heart rate and blood oxygen concentration. In contrast, there is essentially no unique information transfer from blood oxygen concentration, indicating that all of its apparent influence could also be due to heart rate.

We conclude by noting that the multivariate generalization described here extends naturally to any number of sources, simply by applying the general form of PI-decomposition \cite{Williams2010}. The two extensions described in this Letter can also be applied in conjunction, allowing one to quantify, e.g., state-dependent synergistic transfer.  Thus, together these methods provide a completely general framework for characterizing information exchange in complex systems.

\appendix

\section{Proof of Theorem 1}\label{app:proofs}

As defined in \cite{Touchette2004}, a system has \emph{perfect controllability} iff for any initial state $x$ and final state $x'$ there exists a controller state $c$ such that $p(x'|x, c) = 1$.  We first prove that an equivalent definition of perfect controllability is that a system can be moved deterministically to any final state from any distribution of initial states.
\begin{mylma}
A system is perfectly controllable iff for any $x'$ there exists a distribution $p(c|x)$ such that $p(x')=1$ for any distribution $p(x)$.
\end{mylma}
\begin{proof}
If a system is perfectly controllable, we know that for a given $x'$ there exists at least one $c$ for each $x$ such that $p(x'|x, c) = 1$.  Thus, we can choose
\begin{equation}
\operatorname{supp}(C|x) = \{c:p(x'|x,c)=1\}
\end{equation}
for each $x \in X$, which guarantees that $p(x') = 1$ for any distribution $p(x)$.  As this is verified for any $x'$, this proves the direct part of the theorem.

Conversely, note that if $p(x') = 1$ for a given $x'$ and any distribution $p(x)$, it must be that
\begin{equation}
p(x'|x) = \sum_c p(x'|x, c) p(c|x) = 1
\end{equation} 
for each $x \in X$, and thus for each $x$ there must be at least one $c$ for which $p(x'|x, c) = 1$.  As this holds for any $x'$, the converse is proven.
\end{proof}

In order for a system to be perfectly controllable via open-loop control, we also require that the controller acts independently of the initial state, leading to the following definition.
\begin{mydef}
A system has \emph{perfect open-loop controllability} iff for any $x'$ there exists a distribution $p(c|x)$ such that $p(x')=1$ for any distribution $p(x)$, and $I(X;C)=0$.
\end{mydef}

An alternative definition of perfect open-loop controllability is given by the following lemma.

\begin{mylma}
A system has perfect open-loop controllability iff for any $x'$ there exists a $c$ such that $p(x'|c) = 1$.
\end{mylma}
\begin{proof}
If $I(X; C)=0$, then $p(x')$ can be written as
\begin{equation}
p(x') = \sum_c p(c) \sum_x p(x) p(x'|x, c).
\end{equation}
If in addition we have that $p(x')=1$ for a given $x'$ and any distribution $p(x)$, then there must exist a $c$ for which $p(x'|x, c)=1$ for all $x$, i.e., $p(x'|c)=1$.  This holds for any $x'$, so the direct part of the theorem is proven.

Conversely, if for any $x'$ there exists a $c$ such that $p(x'|c)=1$, then for a given $x'$ we can choose
\begin{equation}
\operatorname{supp}(C) = \{c:p(x'|c)=1\}
\end{equation}
with $p(c)=p(c|x)$ for all $c$ and $x$, ensuring that $p(x')=1$ and $I(X;C)=0$.  As this holds for any $x'$, the converse is proven.
\end{proof}

In \cite{Touchette2004}, it is shown that an equivalent information-theoretic definition of perfect controllability is that there exists a distribution $p(c|x)$ such that each final state is reachable from each initial state, i.e.,
\begin{equation}\label{reachability}
p(x'|x) \neq 0 
\end{equation}
for all $x$ and $x'$ and that, for any distribution $p(x)$, $I(X';C|X)$ is maximal, i.e.,
\begin{equation}\label{determ-trans}
H(X'|X, C) = 0
\end{equation} 
so that
\begin{equation}
I(X';C|X) = H(X'|X) - H(X'|X, C) = H(X'|X).
\end{equation} 
Consequently, $I(X'; C|X)$ is naturally interpreted as a system's degree of controllability, which is maximal iff the system is perfectly controllable.

By the same reasoning, Theorem 1 establishes that the SITE from $C$ to $X'$ is a natural measure of a system's open-loop controllability, maximal exactly in the case of perfect open-loop control.  In order to prove Theorem 1, we will need the following basic property of SDTE.
\begin{mylma}
The SDTE from $C$ to $X'$ is zero iff for each $x' \in X'$,
\begin{align*}
p(x'|x, c) & = p(x'|x) \mbox{ } \forall x, c \\
& \mbox{or} \\
p(x'|x, c) & = p(x'|c) \mbox{ } \forall x, c.
\end{align*}
\end{mylma}
\begin{proof}
\begin{align*}
& \Pi_{{\bf R}}(X'; \{ X, C \}) \\
= & I(X'; X, C) - I_{\max}(X'; X, C) \\
= & \sum_{x'} p(x') \Big[I(X'=x'; X, C) - \\
   & \max \{I(X'=x'; X), I(X'=x'; C)\}\Big] \\
= & \sum_{x'} p(x') \min \{ I(X'=x'; C|X), I(X'=x'; X|C) \}
\end{align*}
where
\begin{equation*}
I(X'=x'; C|X) = \sum_x p(x|x') D(p(c|x, x') \parallel p(c|x))
\end{equation*}
and $D(\cdot \parallel \cdot)$ is the Kullback-Leibler divergence \cite{Cover2006}. $D(\cdot \parallel \cdot)$ is nonnegative, so $\Pi_{{\bf R}}(X'; \{ X, C \}) = 0$ iff, for each $x' \in X'$, $I(X'=x'; C|X) = 0$ or $I(X'=x'; X|C) = 0$.  Furthermore, since $D(q \parallel r) = 0$ iff $q=r$, $\Pi_{{\bf R}}(X'; \{ X, C \}) = 0$ iff, for each $x' \in X'$,
\begin{align*}
p(c|x, x') & = p(c|x) \mbox{ } \forall x, c \\
& \mbox{or} \\
p(x|c, x') & = p(x|c) \mbox{ } \forall x, c
\end{align*}
or, equivalently,
\begin{align*}
p(x'|x, c) & = p(x'|x) \mbox{ } \forall x, c \\
& \mbox{or} \\
p(x'|x, c) & = p(x'|c) \mbox{ } \forall x, c.
\end{align*}
\end{proof}

Now we are in a position to prove Theorem 1.

\begin{proof}
We will show that a system has perfect open-loop controllability iff there exists a distribution $p(c|x)$ such that Eqs. \eqref{reachability} and \eqref{determ-trans} are satisfied and there is no SDTE from $C$ to $X'$,
\begin{equation}
\Pi_{{\bf R}}(X'; \{ X, C \}) = 0.
\end{equation}
If a system is open-loop controllable, then for each $x'$ there exists a $c$ such that $p(x'|c) = 1$.  Choosing
\begin{equation}
\operatorname{supp}(C) = \{c:p(x'|c)=1\}
\end{equation}
over all $x' \in X'$, with $p(c)=p(c|x)$ for all $c$ and $x$, ensures that $H(X'|X, C) \leq H(X'|C) = 0$ and $p(x'|x) \neq 0$.  Also, the chosen distribution $p(c|x)$ ensures that $p(x'|x, c) = p(x'|c) \mbox{ } \forall x, c$ so that $\Pi_{{\bf R}}(X'; \{ X, C \}) = 0$ by the preceding lemma.  This proves the direct part of the theorem.

\begin{figure*}[t!]
  \centering
    \includegraphics[width=0.7\textwidth]{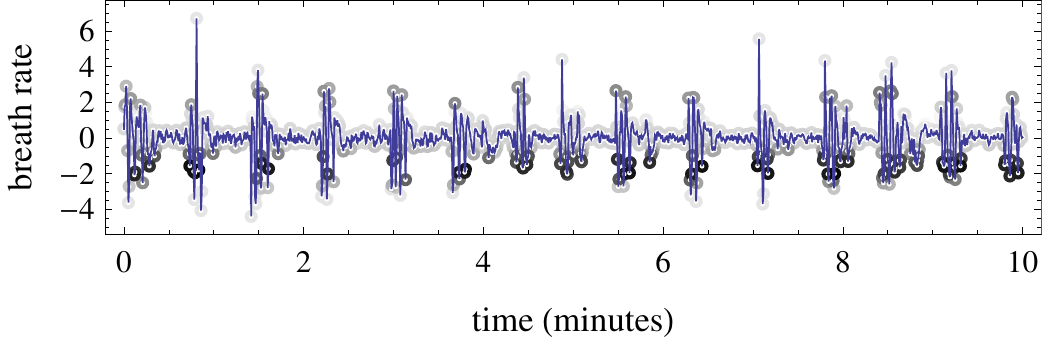}
\caption{Information transfer $I(B_{t+1}; H_t |B_t = b_t)$ from heart rate (H) to breath rate (B) superimposed on the breath rate signal.  Darker colors correspond to higher values of $I(B_{t+1}; H_t |B_t = b_t)$.  Heart rate has the largest influence on breath rate when chest volume is low and, to a lesser extent, when it is high, but minimal influence when chest volume is near its mean.}
\label{supp_physio_sdte_plot}
\end{figure*}

For the converse, note that $p(x'|x) \neq 0$ for a given $x'$ and $x$ means that there is at least one $c$ for which $p(x'|x, c) \neq 0$ and, since $H(X'|X, C) = 0$, we can further conclude that $p(x'|x, c) = 1$.  If we also have that $\Pi_{{\bf R}}(X'; \{ X, C \}) = 0$, we know that, for each $x' \in X'$,
\begin{align*}
p(x'|x, c) & = p(x'|x) \mbox{ } \forall x, c \\
& \mbox{or} \\
p(x'|x, c) & = p(x'|c) \mbox{ } \forall x, c
\end{align*}
and thus that, for each $x' \in X'$, 
\begin{align*}
\exists x & , p(x'|x) = 1 \\
& \mbox{or} \\
\exists c &, p(x'|c) = 1.
\end{align*}
But it cannot be the case that $\exists x, p(x'|x) = 1$, since that would violate the reachability condition (Eq. \eqref{reachability}), so we conclude that for each $x'$ there exists a $c$ such that $p(x'|c) = 1$.  This proves the converse.
\end{proof}

\section{State-Dependent Influence on Breath Rate}\label{app:fig}

As mentioned in the main text, we found that heart rate has an exclusively state-dependent influence on breath rate.  This is shown most clearly by superimposing the information transfer values on the breath rate signal (FIG.~\ref{supp_physio_sdte_plot}).

\begin{acknowledgments}
We thank J. Beggs, V. Griffith, A. Kolchinsky, and O. Sporns.  This work was supported in part by NSF grant IIS-0916409 (to R.D.B.) and an IGERT traineeship (to P.L.W.).
\end{acknowledgments}

\bibliography{pi+te_arxiv}

\end{document}